\documentclass{amsart}
\usepackage{times, stmaryrd, verbatim}
\usepackage{ amssymb, mathrsfs} %from AMS

\newtheorem{Thm}{Theorem}

\newtheorem{Prop}{Proposition}[section]

\newtheorem{rmk}{Remark}[section]

\numberwithin{equation}{section}

\newcommand{\mr}{\mathrm}

\newcommand{\tr}{\mr{tr}\,}
\newcommand{\mrd}{\mathrm{d}}

\usepackage{hyperref}
\begin{document}

\title{On the trajectories of $\mr O(1)$-Kepler Problems}
\author{Guowu Meng}

\address{Department of Mathematics, Hong Kong Univ. of Sci. and
Tech., Clear Water Bay, Kowloon, Hong Kong}
%\curraddr{Institute for Advanced Study, Einstein Drive, Princeton, New Jersey 08540 USA}

\email{mameng@ust.hk}
\thanks{The author was supported by the Hong Hong Research Grants Council under RGC Project No. 16304014 and the Hong Kong University of Science and Technology under DAG S09/10.SC02.}

%    General info
%\subjclass[2000]{Primary 22E46, 22E70; Secondary 81S99, 51P05}

\date{\today}


\maketitle
\begin{abstract}
The trajectories of the $\mathrm{O}(1)$-Kepler problem at level $n\ge 2$ are completely determined. It is found in particular that a non-colliding trajectory is an ellipse, a parabola or a branch of hyperbola according as the total energy is negative, zero or positive. Moreover, it is shown that the group $\mathrm{GL}(n, \mathbb R)/\mathrm{O}(1)$ acts transitively on both the set of oriented elliptic trajectories and the set of oriented parabolic trajectories. The method employed here is similar to the one used by Levi-Civita in the study of planar Kepler problem in 1920.
\end{abstract}

%\tableofcontents

\section {Introduction}
The classical $\mathrm{O}(1)$-Kepler problem at level $n\ge 2$ is a generalized Kepler problem \cite{meng2011} whose configuration space is the space $\mathcal C_1$ of rank one semi-positive elements in the euclidean Jordan algebra $\mathrm H_n(\mathbb R)$ of real symmetric matrices of order $n$, with Lagrangian function
\begin{eqnarray}\label{Lagrangian}
L={1\over 2}||\dot x||^2+{n\over \tr x}.
\end{eqnarray} 
Here $\tr x$ is the trace of $x$, hence always positive because $x\in \mathcal C_1$. The length square of the velocity vector $\dot x$ is not calculated with the euclidean structure on $\mathrm H_n(\mathbb R)$. To describe it, we note that the tangent space $T_x\mathcal C_1$ is the subspace $\{x\}\times \mathrm{Range}\, L_x$ of $T_x\mathrm{H}_n(\mathbb R)$ where $L_x$ is the Jordan multiplication by $x$, i.e., $L_xy={1\over 2}(xy+yx)$ --- the symmetrized matrix product of $x$ with $y$. We also note that $L_x$: $\mathrm{H}_n(\mathbb R)\to \mathrm{H}_n(\mathbb R)$ maps $\mathrm{Range}\, L_x$ isomorphically onto $\mathrm{Range}\, L_x$, so we have an automorphism $\bar L_x$ of $T_x\mathcal C_1$. (Both of these two statements can be easily verified by assuming that $x$ is in diagonal form.) By definition
\begin{eqnarray}
||\dot x||^2= {1\over n}\tr x\; \langle \dot x, \bar L_x^{-1} (\dot x) \rangle
\end{eqnarray} where $\bar L_x^{-1} $ is the inverse of $\bar L_x$, and $\langle\, , \, \rangle$ is the inner product on $T_x\mathrm{H}_n(\mathbb R)$: for $(x, u)$, $(x, v)$ in $T_x\mathrm{H}_n(\mathbb R)$, we have 
\begin{eqnarray}
\langle (x, u), (x, v) \rangle = {1\over n}\tr (uv)
\end{eqnarray} 
where $uv$ is the matrix (or Jordan) multiplication of $u$ with $v$. Note that the $\mathrm{O}(1)$-Kepler problem at level $2$ is just the planar Kepler problem.

\vskip 10pt
The bound state analysis for $\mr O(1)$-Kepler problems has been done in Ref. \cite{meng2008}. Here we analyze the trajectories of $\mr O(1)$-Kepler problems. We shall show that a trajectory is always the intersection of $\mathcal C_1$ with a real plane inside $\mathrm{H}_n(\mathbb R)$, consequently, since 
\[
\mathcal C_1 =\{x\in \mathrm{H}_n(\mathbb R)\mid x^2=\tr x\, x, \; \tr x>0\}
\]
a trajectory must be a quadratic plane curve. In fact, we shall show that a non-colliding trajectory is an ellipse, a parabola or a branch of a hyperbola according as the total energy 
\begin{eqnarray}
E={1\over 2}||\dot x||^2-{n\over \tr x}
\end{eqnarray} is negative, zero or positive. It will also be shown that the group $\mathrm{GL}(n, \mathbb R)/\mathrm{O}(1)$ acts transitively on both the set of elliptic trajectories and the set of parabolic trajectories. (Here $\mathrm{GL}(n, \mathbb R)/\mathrm{O}(1)$ is the quotient group of $\mathrm{GL}(n, \mathbb R)$ by the normal subgroup consisting of the identity matrix of order $n$ and its negative.) Note that, a trajectory is the path traced by a motion, so it is oriented by the velocity of the motion. 

\begin{rmk}
The analogue of Laplace-Runge-Lenz vector for the classical $\mathrm{O}(1)$-Kepler problem at level $n$ exists and can be obtained via a natural Poisson representation of $\mathfrak{sp}_{2n}(\mathbb R)$ on $T^*\mathcal C_1$. For more details, please consult \cite[Subsection~4.1]{meng2014}.
\end{rmk}

\subsection{Notations} If $u$, $v$ are vectors in $\mathbb R^n$, we use $u\cdot v$ to denote the dot product of $u$ with $v$, $u^2$ to denote $u\cdot u$. If $x$ is a square matrix, we use $\tr x$ to denote the trace of $x$, and $x^2$ to denote the matrix multiplication of $x$ with $x$. 

\section{Solving equation of motion}
The equation of motion for the planar Kepler problem was ingeniously solved in Ref. \cite{Levi-Civita1920} by transforming the nonlinear equation of motion into a linear ordinary differential equation (ODE). This transformation, referred to as the \emph{Levi-Civita transformation} in the literature, is based on the quadratic map: $z\in \mathbb C \mapsto z^2\in \mathbb C$.

We shall use an analogous idea to solve the equation of motion for $\mr O(1)$-Kepler problems. The analogous transformation that we shall use, which turns the equation of motion into a linear ODE, is based on the following quadratic map
\begin{eqnarray}
\begin{array}{lccc}
q: & \mathbb R^n &\to & \mathrm H_n(\mathbb R)\\
& X&\mapsto & n\,XX^{\mathrm t}
\end{array}
\end{eqnarray} where $X^{\mathrm t}$ is the transpose of the column vector $X$ and $XX^{\mathrm t}$ is the matrix multiplication of $X$ with $X^{\mathrm t}$. The mapping $q$, when restricted to $\mathbb R^n_*:=\mathbb R^n\setminus\{\mathbf 0\}$, becomes a two-to-one covering map onto $\mathcal C_1$:
\begin{eqnarray}\label{quadraticmap}
\bar q: \quad \mathbb R^n_*\buildrel 2:1\over\longrightarrow \mathcal C_1.
\end{eqnarray}
Consequently the tangent map $T\bar q$:\, $T\mathbb R^n_*\to T\mathcal C_1$ is also a two-to-one covering map.

\begin{Prop}\label{Ltilde}
The composition of Lagrangian $L$ in Eq. \eqref{Lagrangian} with the tangent map $T\bar q$, denoted by $\tilde L$:  $T\mathbb R^n_*\to \mathbb R$, 
is
\[
\tilde L =  2X^2\dot X^2+{1\over X^2}
\] where $X^2=X\cdot X$ and $\dot X^2=\dot X\cdot \dot X$. Also, $\tilde E:=E\circ T\bar q$: $T\mathbb R^n_*\to \mathbb R$ is given by
\[
\tilde E =  2X^2\dot X^2-{1\over X^2}.
\] 
\end{Prop}
\begin{proof}
It is clear that $\tr (XX^{\mathrm t})=X^2$. Thus it suffices to show that 
\begin{eqnarray}\label{velocityidentity}
\tr\left( {\mrd\over \mrd t}(XX^{\mathrm t})\; \bar L_{XX^{\mathrm t}}^{-1} \left( {\mrd\over \mrd t}(XX^{\mathrm t})\right)\right) = 4 \dot X^2.
\end{eqnarray}
Since $\bar q$ is $\mathrm{O}(n)$-equivariant, and both $L$ and $\tilde L$ are $\mathrm{O}(n)$-invariant, it suffices to verify the proposition at the point $X=[a, 0, \ldots, 0]^{\mathrm t}$ with $a>0$. Let $\dot X=[y_1, \ldots, y_n]^{\mathrm t}$. Then

\renewcommand{\arraystretch}{2}
$$\newcommand*{\temp}{\multicolumn{1}{r|}{}}
XX^{\mathrm t}=a^2\left[\begin{array}{ccccc}
1 &\temp & 0&\cdots &0 \\ \cline{1-5}
0 &\temp &0& \cdots & 0\\  
\vdots &\temp & \vdots&\ddots &\vdots\\
0 &\temp & 0&\cdots &0
\end{array}\right],
$$
\renewcommand{\arraystretch}{2}
$$\newcommand*{\temp}{\multicolumn{1}{r|}{}}
{\mrd\over \mrd t}(XX^{\mathrm t})=a\left[\begin{array}{ccccc}
2y_1 &\temp & y_2&\cdots &y_n \\ \cline{1-5}
y_2 &\temp &0& \cdots & 0\\  
\vdots &\temp & \vdots&\ddots &\vdots\\
y_n &\temp & 0&\cdots &0
\end{array}\right]
$$
and
\renewcommand{\arraystretch}{2}
$$\newcommand*{\temp}{\multicolumn{1}{r|}{}}
\bar L_{XX^{\mathrm t}}^{-1}\left({\mrd\over \mrd t}(XX^{\mathrm t})\right)={1\over a}\left[\begin{array}{ccccc}
2y_1 &\temp & 2y_2&\cdots &2y_n \\ \cline{1-5}
2y_2 &\temp &0& \cdots & 0\\  
\vdots &\temp & \vdots&\ddots &\vdots\\
2y_n &\temp & 0&\cdots &0
\end{array}\right].
$$
Then we have 
\[
\tr\left( {\mrd\over \mrd t}(XX^{\mathrm t})\; \bar L_{XX^{\mathrm t}}^{-1} \left( {\mrd\over \mrd t}(XX^{\mathrm t})\right)\right) =4\sum_{i=1}^n y_i^2 = 4\dot X^2.
\]
\end{proof}
\begin{rmk}
The dynamical problem with configuration space $\mathbb R^n_*$ and Lagrangian $\tilde L$ in Proposition \ref{Ltilde} is a conformal Kepler problem in the sense of T. Iwai \cite{Iwai1981}. 
\end{rmk}
In view of the fact that a smooth covering map is a local diffeomorphism, the following proposition, one of the two reasons for the success of Levi-Civita's approach to planar Kepler problem, is almost evident.
\begin{Prop}\label{covering}
Let $p$: $E\to X$ be a covering map from manifold $E$ onto manifold $X$, \\$L$: $TX\to \mathbb R$ be a smooth function, $\alpha$: $I\to X$ be a smooth map from interval $I$ to $X$. Then $\alpha$ is a solution to the Euler-Lagrange equation associated with Lagrangian $L$ if and only if any lifting $\tilde \alpha$ of $\alpha$ is a solution to the Euler-Lagrange equation associated with Lagrangian $L\circ Tp$. 
\end{Prop}
\begin{rmk}
The existence of such a lifting is almost the definition of covering space. 
\end{rmk}
In view of Proposition \ref{covering}, by finding all solutions to the Euler-Lagrange equation \cite[p. 58]{Arnold89} associated with Lagrange $\tilde L$ in Proposition \ref{Ltilde} and then composing them with $\bar q$, we get all solutions to the equation of motion for the $\mathrm{O}(1)$-Kepler problem at level $n$.

From Proposition \ref{Ltilde} one can see that the Euler-Lagrange equation associated with $\tilde L$ is
$\displaystyle{\mrd\over \mrd t}(4X^2\dot X)=\tilde E{2X\over X^2}$ or
\begin{eqnarray}\label{EqnOfM}
\left(X^2{\mrd\over \mrd t}\right)^2 X={\tilde E\over 2}\, X.
\end{eqnarray}

Let us fix a solution to the equation of motion for the $\mathrm{O}(1)$-Kepler problem at level $n$, then the total energy $E$ is a constant of motion, and so is $\tilde E=E$. With this in mind we can solve Eq. \eqref{EqnOfM} in three cases according as the total energy 
$E$ is negative, zero or positive. 

\subsection{The case $E<0$}
In this case we introduce variable
\[
\tau =\sqrt{-E\over 2}\displaystyle \int_0^t {\mrd\tilde t\over X(\tilde t)^2}.
\]
Then $\tau$ is an increasing smooth function of $t$ and Eq. \eqref{EqnOfM} becomes
\[
{\mrd^2X\over \mrd\tau^2}+X=0,
\] so $X$ is of the form
\begin{eqnarray}\label{solnegE}
X(t(\tau))=\cos \tau\, u+\sin \tau\, v
\end{eqnarray}  for some $u\in \mathbb R^n_*$ and $v\in \mathbb R^n$. Substituting this solution to equation $E =  2X^2\dot X^2-{1\over X^2}$, we get
\[
E=-{1\over u^2+v^2}.
\]
Although $\tau $ is a complicated function of $t$, $t$ can be quiet explicitly expressed as a function of $\tau$:
\begin{eqnarray}
t&=&\sqrt{2\over -E}\displaystyle \int_0^\tau X(t(\tilde\tau))^2 \, \mrd\tilde \tau\cr
&=&\sqrt{2\over -E}\displaystyle \int_0^\tau \left( \cos \tilde\tau\, u+\sin \tilde\tau\, v
   \right)^2 \, \mrd\tilde \tau\cr
   &=&\sqrt{2(u^2+v^2)}\left({u^2+v^2\over 2}\tau+{u^2-v^2\over 4}\sin(2\tau)+{u\cdot v\over 2}\left(1-\cos(2\tau)\right)\right).\nonumber
\end{eqnarray}

\subsection{The case $E=0$}
In this case we introduce variable
\[
\tau =\displaystyle \int_0^t {\mrd s\over X(s)^2}.
\]
Then $\tau$ is an increasing smooth function of $t$ and Eq. \eqref{EqnOfM} becomes
\[
{\mr d^2X\over \mrd\tau^2}=0,
\] so $X$ is of the form
\begin{eqnarray}\label{solnegE}
X(t(\tau))=u+\tau v
\end{eqnarray}  for some $u\in \mathbb R^n_*$ and $v\in \mathbb R^n$. Substituting this solution to equation $0=E =  2X^2\dot X^2-{1\over X^2}$, we get
$v^2={1\over 2}$. Again $t$ is a simple increasing function of $\tau$:
\[
t=u^2\tau+u\cdot v\, \tau^2+{1\over 6}\tau^3.
\]

\subsection{The case $E>0$}
In this case we introduce variable
\[
\tau =\sqrt{E\over 2}\displaystyle \int_0^t {\mrd s\over X(s)^2}.
\]
Then $\tau$ is an increasing smooth function of $t$ and Eq. \eqref{EqnOfM} becomes
\[
{\mr d^2X\over \mrd\tau^2}-X=0,
\] so $X$ is of the form
\begin{eqnarray}\label{solposE}
X(t(\tau))=\cosh \tau\, u+\sinh \tau\, v
\end{eqnarray} for some $u\in \mathbb R^n_*$ and $v\in \mathbb R^n$. Substituting this solution to equation $E =  2X^2\dot X^2-{1\over X^2}$, we get
\[
E={1\over v^2-u^2}.
\] Since $E>0$, we must have $v^2>u^2$ in solution \eqref{solposE}. Again, $t$ is a simple increasing function of $\tau$:
\begin{eqnarray}
t &=&\sqrt{2(v^2-u^2)}\left({u^2-v^2\over 2}\tau+{u^2+v^2\over 4}\sinh(2\tau)+{u\cdot v\over 2}\left(\cosh(2\tau)-1\right)\right).\nonumber
\end{eqnarray}

\vskip 10pt
The above analysis, when combined with Proposition \ref{covering}, yields all solutions to the equation of motion of the $\mr O(1)$-Kepler problem at level $n$, though the dependence on time $t$ is only implicitly given. Moreover, for any solution $X(t)$ to Eq. \eqref{EqnOfM} we have obtained above, one can check that the image of the map $t\mapsto q(X(t))$ always lies inside a real plane inside $\mathrm{H}_n(\mathbb R)$. Therefore, in combination with Proposition \ref{covering}, the above analysis implies

\begin{Thm} \label{trajectory}For the $\mathrm{O}(1)$-Kepler problem at level $n$, the following statements are true.

1) A trajectory is always the intersection of the space $\mathcal C_1$ with a real plane inside $\mathrm{H}_n(\mathbb R)$, and it is bounded or unbounded according as the total energy $E$ is negative or not.

2) A bounded trajectory can be parametrized as $\alpha(\tau)=q(\cos \tau\, u+\sin \tau\, v)$ for some $u\in \mathbb R_*^n$ and $v\in \mathbb R^n$. Moreover, any parametrized curve of this form is a bounded trajectory
with negative total energy $E=-{1\over u^2+v^2}$. 

3) An unbounded trajectory with zero total energy can be parametrized as $\alpha(\tau)=q(u+ v\tau)$ for some $u\in \mathbb R_*^n$ and $v\in \mathbb R^n$ with $v^2={1\over 2}$. Moreover, any parametrized curve of this form is a trajectory with zero total energy. 

4) An unbounded trajectory with positive total energy can be parametrized as $\alpha(\tau)=q(\cosh \tau\, u+\sinh \tau\, v)$ for some $u\in \mathbb R_*^n$ and $v\in \mathbb R^n$ with $v^2>u^2$. Moreover, any parametrized curve of this form is a trajectory with positive total energy $E={1\over v^2-u^2}$.

\end{Thm}

\section{Non-colliding trajectories}
The interesting trajectories are the non-colliding ones, i.e., the ones such that in their parametrization $\alpha(\tau)$ given in theorem \ref{trajectory}, $\alpha(\tau)\neq \mathbf 0\in \mathrm{H}_n(\mathbb R)$ for any $\tau\in\mathbb R$. It is evident that if $v$ is a scalar multiple of $u$ in theorem \ref{trajectory}, then $\alpha(\tau)=\mathbf 0$ for some finite value of $\tau$ and it is not hard to check that the converse is also true.  
Therefore, applied to non-colliding trajectories only, theorem \ref{trajectory}  becomes 

\begin{Thm}\label{non-colliding trajectories} For a non-colliding trajectory of the $\mathrm{O}(1)$-Kepler problem at level $n$, the following statements are true.

1) It is an ellipse, a parabola or a branch of hyperbola according as the total energy $E$ is negative, zero or positive.

(We assume in the next three statements that the variable $\tau$ runs over the entire $\mathbb R$.)

2) If it is an ellipse then it can be parametrized as $\alpha(\tau)=q(\cos \tau\, u+\sin \tau\, v)$ for some linearly independent $u, v\in \mathbb R^n$.  Moreover, any parametrized curve of this form is an elliptic trajectory with negative total energy $E=-{1\over u^2+v^2}$. 

3)  If it is a parabola then it can be parametrized as $\alpha(\tau)=q(u+ v\tau)$ for some linearly independent $u, v\in \mathbb R^n$. Moreover, any parametrized curve of this form is a parabolic trajectory with zero total energy. 

4)  If it is a branch of hyperbola then it can be parametrized as $\alpha(\tau)=q(\cosh \tau\, u+\sinh \tau\, v)$ for some linearly independent $u, v\in \mathbb R^n$ with $v^2>u^2$. Moreover, any parametrized curve of this form is a hyperbolic trajectory with positive total energy $E={1\over v^2-u^2}$. 
\end{Thm}
Note that, in statement 3) of Theorem \ref{non-colliding trajectories} the condition $v^2={1\over 2}$ is no longer needed because one can rescale $v$ due to the fact that $\tau\in \mathbb R$. Let $\mathrm{GL}(n, \mathbb R)/\mathrm{O}(1)$ be the quotient group of $\mathrm{GL}(n, \mathbb R)$ by the normal subgroup consisting of the identity matrix of order $n$ and its negative. Since the standard linear action of $\mathrm{GL}(n, \mathbb R)$ on $\mathbb R^n$ ($n\ge 2$) acts transitively on the set of linearly independent pairs of vectors in $\mathbb R^n$, Theorem \ref{non-colliding trajectories} implies the following
\begin{Thm}
For the $\mathrm{O}(1)$-Kepler problem at level $n$, the group $\mathrm{GL}(n, \mathbb R)/\mathrm{O}(1)$ acts transitively on both set of oriented elliptic  trajectories and the set of oriented parabolic trajectories. 
\end{Thm} 
This theorem is a direct analogue of parts 3) and 4) in Theorem 2 of Ref. \cite{meng2012}.  Note that, when and only when $n$ is odd, $\mathrm{GL}(n, \mathbb R)/\mathrm{O}(1)$ is the orientation-preserving linear automorphism group of $\mathbb R^n$.

\end{document}